\documentclass[submission,copyright,creativecommons]{eptcs}
\providecommand{\event}{LFMTP 2015} 
\usepackage{amsthm}
\usepackage{amsmath}
\usepackage{amssymb}
\usepackage{mdframed}
\usepackage{bussproofs}
\usepackage{multicol}
\usepackage{tikz}

\newtheorem{theorem}{Theorem}[section]
\newtheorem{lemma}{Lemma}[section]

\newtheorem{definition}{Definition}[section]

\newtheorem{notation}{Notation}[section]

\newcommand\lpm{$\la\Pi$-calculus Modulo}

\newcommand\AxC[1]{\AxiomC{#1}}
\newcommand\UIC[1]{\UnaryInfC{#1}}
\newcommand\BIC[1]{\BinaryInfC{#1}}
\newcommand\TIC[1]{\TrinaryInfC{#1}}
\newcommand\QIC[1]{\QuaternaryInfC{#1}}
\newcommand{\DP}{\DisplayProof}

\newcommand\ud[1]{\underline{#1}}

\newcommand\la{\lambda}

\newcommand\be{\beta}
\newcommand\si{\sigma}
\newcommand\De{\Delta}
\newcommand\Ga{\Gamma}
\newcommand\Si{\Sigma}

\newcommand\Type{{\bf Type}}
\newcommand\Kind{{\bf Kind}}

\newcommand\wf[1]{#1\ {\bf wf}}
\newcommand\typctx[1]{\Ga\vdash^{ctx}#1}
\newcommand\typ[4]{#1;#2\vdash#3:#4}
\newcommand\typg[3]{\Ga;#1\vdash#2:#3}

\newcommand\typrr[2]{\Ga\vdash#1\rw#2}

\newcommand\rw{\hookrightarrow}
\newcommand\eqbg{\equiv_{\be\Ga}}
\newcommand\eqb{\equiv_{\be}}

\newcommand\OCst{\mathcal{C}_O}
\newcommand\TCst{\mathcal{C}_T}
\newcommand\Var{\mathcal{V}}

\newcommand\Sig{{\bf Sig}}
\newcommand\ar{\longrightarrow}

\newcommand\redbg{\rightarrow_{\be\Ga}}
\newcommand\redg{\rightarrow_{\Ga}}
\newcommand\redb{\rightarrow_{\be}}

\newcommand\redbgb{\rightarrow_{\be\Ga^{b}}}
\newcommand\redgb{\rightarrow_{\Ga^{b}}}
\newcommand\eqbgb{\equiv_{\be\Ga^{b}}}

\newcommand\Nat{{\tt Nat}}
\newcommand\Succ{{\tt S}}
\newcommand\Plus{{\tt plus}}
\newcommand\List{{\tt List}}
\newcommand\Map{{\tt Map}}
\newcommand\Nil{{\tt Nil}}
\newcommand\Cons{{\tt Cons}}

\newcommand\Diff{{\tt D}}

\newcommand\Exp{{\tt Exp}}
\newcommand\Fmult{{\tt fMult}}

\newcommand\Base{\mathcal{B}}

\newcommand\ula{\ud{\la}}
\newcommand\bel{\updownarrow^\eta_\be}
\newcommand\red{\rightarrow}

\newcommand\hrsTerm{{\tt Term}}
\newcommand\hrsType{{\tt Type}}
\newcommand\hrsKind{{\tt Kind}}
\newcommand\hrsApp{{\tt App}}
\newcommand\hrsLam{{\tt Lam}}
\newcommand\hrsPi{{\tt Pi}}

\newcommand\emb[1]{\Vert#1\Vert}
\newcommand\embfv[2]{\Vert#1\Vert_{#2}}

\title{Rewriting Modulo $\be$ in the $\la\Pi$-Calculus Modulo}

\author{Ronan Saillard
	\institute{MINES ParisTech, PSL Research University, France}
	\email{ronan.saillard@mines-paristech.fr}
}
\def\titlerunning{Rewriting Modulo $\be$ in the \lpm}
\def\authorrunning{R. Saillard}
\begin{document}
\maketitle

\begin{abstract}
	The \lpm{} is a variant of the $\la$-calculus with dependent types
	where $\be$-conversion is extended with user-defined rewrite rules.
	It is an expressive logical framework and has been used to encode
	logics and type systems in a shallow way.
	Basic properties such as subject reduction or uniqueness of types do
	not hold in general in the \lpm{}. However, they hold if the rewrite
	system generated by the rewrite rules together with $\be$-reduction
	is confluent.
	But this is too restrictive.
	To handle the case where non confluence comes from the interference
	between the $\be$-reduction and rewrite rules with $\la$-abstraction
	on their left-hand side, we introduce a notion of rewriting modulo $\be$ for the \lpm{}.
	We prove that confluence of rewriting modulo $\be$ is enough to ensure subject
	reduction and uniqueness of types.
	We achieve our goal by encoding the \lpm{} into Higher-Order Rewrite System (HRS).
	As a consequence, we also make the confluence results for HRSs available for the \lpm{}.
\end{abstract}

\section{Introduction}
\label{sec:introduction}

The \lpm{} is a variant of the $\la$-calculus with dependent types ($\la\Pi$-calculus or LF)
where $\be$-conversion is extended with user-defined rewrite rules.
Since its introduction by Cousineau and Dowek~\cite{CousineauDowek07},
it has been used as a logical framework to express different logics and type systems.
A key advantage of rewrite rules is that they allow designing {\em shallow}\/ embeddings,
that is embeddings that preserve the computational content of the encoded system.
It has been used, for instance, to encode functional Pure Type Systems~\cite{CousineauDowek07},
First-Order Logic~\cite{DorraFO},
Higher-Order Logic~\cite{AliHOL},
the Calculus of Inductive Constructions~\cite{Coqine},
resolution and superposition proofs~\cite{Resolution},
and the $\varsigma$-calculus~\cite{RaphaelSigma}.

The expressive power of the \lpm{} comes at a cost:
basic properties such as subject reduction or uniqueness of types
do not hold in general.
Therefore, one has to prove these properties for each particular set of rewrite rules considered.
The usual way to do so is to prove that the rewriting relation generated by the
rewrite rules together with $\be$-reduction is confluent.
This entails a property called product compatibility (also known as $\Pi$-injectivity
or injectivity of function types) which,
in turn, implies both subject reduction and uniqueness of types.
Another important consequence of confluence is that, together with termination,
it implies the decidability of the corresponding congruence.
Indeed, for confluent and terminating relations, checking congruence boils down
to a syntactic equality check between normal forms.
As a direct corollary, we get the decidability of type checking in the \lpm{}
for the corresponding rewrite relations.

One case where confluence is easily lost is if one allows rewrite rules with $\la$-abstractions
on their left-hand side.
For instance, consider the following rewrite rule
(which reflects the mathematical equality $(e^f)' = f'*e^f$):

$$ \Diff~(\la x:R.\Exp~(f~x)) \rw \Fmult~(\Diff~(\la x:R.f~x))~(\la x:R.\Exp~(f~x)). $$

This rule introduces a non-joinable critical peak when combined with $\be$-reduction:
\\[0.5cm]
\begin{tikzpicture}
	\node (n1) at (8,2) {$\Diff~(\la x:R.\Exp~((\la y:R.y)~x))$};
	\node (n2) at (4,0) {$\Fmult~(\Diff~(\la x:R.(\la y:R.y)~x))~(\la x:R.(\Exp~((\la y:R.y)~x)))$};
	\node (n3) at (12,0) {$\Diff~(\la x:R.\Exp~x)$};

	\draw[->] (n1) -- (n2) node[pos=0.7,left]{$\begin{array}{l}\Diff\\ \end{array}$};
	\draw[->] (n1) -- (n3) node[pos=0.7,right]{$\begin{array}{l}\be\\ \end{array}$};
\end{tikzpicture}

A way to recover confluence is to consider a generalized rewriting relation where matching
is done modulo $\be$-reduction.
In this setting $\Diff~(\la x:R.\Exp~x)$ is reducible because it is $\be$-equivalent
to the redex $\Diff~(\la x:R.\Exp ((\la y:R.y)~x))$ and, as we will see, this allows
closing the critical peak.

In this paper, we formalize the notion of {\em rewriting modulo} $\be$ in the context of the \lpm.
We achieve this by encoding the \lpm{} into Nipkow's Higher-Order Rewrite Systems~\cite{Nipkow-LICS-91}.
This encoding allows us,
first, to properly define matching modulo $\be$ using the notion of higher order rewriting and,
secondly, to make available, in the \lpm{}, confluence and termination criteria designed for
higher-order rewriting.
Then we prove that the assumption of confluence for the rewriting modulo $\be$ relation can be used,
in most proofs, in place of standard confluence.
In particular this implies subject reduction (for both standard rewriting and rewriting modulo $\be$)
and uniqueness of types.

The paper is organized as follows.
First, we define in \autoref{sec:lpm} the $\la\Pi$-calculus modulo for which we prove subject reduction
and uniqueness of types under the assumption of product compatibility and we show that confluence implies this latter property.
In \autoref{sec:naive-beta-rewriting}, we show that a naive definition of rewriting modulo $\be$
does not work in a typed setting.
This leads us to use Higher-Order Rewrite Systems which we present in \autoref{sec:higher-order-rewrite-systems}
and in which we encode the \lpm{} in \autoref{sec:encoding}.
Then, we use this encoding to properly define rewriting modulo $\be$ in \autoref{sec:rewriting-modulo-beta}
and generalize the results of the previous sections.
We discuss possible applications in~\autoref{sec:Applications} before concluding in~\autoref{sec:conclusion}.

\section{The \texorpdfstring{$\la\Pi$}{}-Calculus Modulo}
\label{sec:lpm}

The \lpm{} is an extension of the dependently-typed $\la$-calculus ($\la\Pi$-calculus)
where the $\be$-conversion is extended by user-defined rewrite rules.

\subsection{Terms}
\label{ssec:lpm-terms}

The terms of the \lpm{} are the same as the terms of the $\la\Pi$-calculus.
Their syntax is given in \autoref{fig:syntax-term}.

\begin{figure}
	\begin{mdframed}
		\centering
		\begin{tabular}{llll}
			$x, y, z$		&	$\in$	&	$\Var$ 	& (Variable) \\
			$\ud{c},\ud{f}$ 			&	$\in$	&	$\OCst$ 	& (Object Constant) \\
			$C, F$ 			&	$\in$	&	$\TCst$		& (Type Constant) \\
			$\ud{t},\ud{u},\ud{v}$		&::=	&
			$x~|~\ud{c}~|~\ud{u}~\ud{v}~|~\la x:U.\ud{t}$ & (Object) \\
			$U,V$						&::=	&
			$C~|~U~\ud{v}~|~\la x:U.V~|~\Pi x:U.V$ & (Type)\\
			$K$							&::=	&
			$\Type~|~\Pi x:U.K$ & (Kind) \\
			$t,u,v$							&::=	&
			$\ud{u}~|~U~|~K~|~\Kind$	& (Term)
		\end{tabular}
		\caption{The terms of the \lpm{}}\label{fig:syntax-terms}
		\label{fig:syntax-term}
	\end{mdframed}
\end{figure}

\begin{definition}[Object, Type, Kind, Term]\label{def:term}
	A \emph{term} is either an \emph{object}, a \emph{type}, a \emph{kind} or the symbol $\Kind$.

	An object is
	either a \emph{variable} in the set $\Var$,
	or an \emph{object constant} in the set $\OCst$,
	or an \emph{application} $\ud{u}$ $\ud{v}$ of two objects,
	or an \emph{abstraction} $\la x:A.\ud{t}$ where $A$ is a type and $\ud{t}$ is an object.

	A type is
	either a \emph{type constant} in the set $\TCst$,
	or an \emph{application} $U$ $\ud{v}$ where $U$ is a type and $\ud{v}$ is an object,
	or an \emph{abstraction} $\la x:U.V$ where $U$ and $V$ are types,
	or a \emph{product} $\Pi x:U.V$ where $U$ and $V$ are types.

	A kind is either a \emph{product} $\Pi x:U.K$ where $U$ is a type and $K$ is a kind
	or the symbol $\Type$.

    $\Type$ and $\Kind$ are called \emph{sorts}.

	The sets $\Var$, $\OCst$ and $\TCst$ are assumed to be infinite and pairwise disjoint.
\end{definition}

\begin{definition}
	A term is \emph{algebraic} if it is not a variable, it is built from constants, variables and applications
	and variables do not have arguments.
\end{definition}


\begin{notation}
	In addition to the naming convention of \autoref{fig:syntax-terms}, we use
	$A$ and $B$ to denote types or kinds;
	$T$ to denote a type, a kind or $\Kind$;
	$s$ for $\Type$ or $\Kind$.

	Moreover, we write $t\vec{u}$ to denote the application of $t$ to an arbitrary number of
	arguments $u_1, \ldots, u_n$.
	We write $u[x/v]$ for the usual (capture-avoiding) substitution of $x$ by $v$ in $u$.
	We write $A \ar B$ for $\Pi x:A.B$ when $B$ does not depend on $x$.
\end{notation}

\subsection{Contexts}
\label{ssec:lpm-contexts}

We distinguish two kinds of context: local and global contexts.
A local context is a list of typing declarations corresponding to variables.
The syntax for contexts is given in~\autoref{fig:syntax-contexts-lpm}.

\begin{figure}
	\begin{mdframed}
		\centering
		\begin{tabular}{llll}
			$\De$ &::= & $\emptyset~|~\De(x:U)$	& (Local Context) \\
			$\Ga$ &::= & $\emptyset~|~\Ga(\ud{c}:U)~|~\Ga(C:K)~|~\Ga(\ud{u}\rw \ud{v})~|~\Ga(U\rw V)$	& (Global Context)
		\end{tabular}
		\caption{Syntax for contexts}\label{fig:syntax-contexts-lpm}
	\end{mdframed}
\end{figure}

\begin{definition}[Local Context]\label{def:lpm-local-context}
	A \emph{local context} is a list of variable declarations (variables together with their type).
\end{definition}

Following our previous work~\cite{IWIL13}, we give a presentation of the \lpm{}
where the rewrite rules are internalized in the system as part of the global context.
This is a difference with earlier presentations~\cite{CousineauDowek07} where the rewrite
rules lived \emph{outside} the system and were typed in an external system
(either the simply-typed calculus or the $\la\Pi$-calculus).
The main benefit of this approach is that the typing of the rewrite rules is made explicit
and becomes an iterative process: rewrite rules previously added in the system can be
used to type new ones.

\begin{definition}
    A \emph{rewrite rule} is a pair of terms.
    We distinguish \emph{object-level rewrite rules} (pairs of objects)
    from \emph{type-level rewrite rules} (pairs of types).

    These are the only allowed rewrite rules.
    We write $(u\rw v)$ for the rewrite rule $(u,v)$.

    It is \emph{left-algebraic} if $u$ is algebraic
    and \emph{left-linear} if no free variable occurs twice in $u$.
\end{definition}

\begin{definition}[Global Context]\label{def:global-context}
	A \emph{global context} is a list of object declarations (an object constant together with a type),
	type declarations (a type constant together with a kind), object-level rewrite rules
	and type-level rewrite rules.
\end{definition}

\subsection{Rewriting}
\label{ssec:lpm-rewriting}

\begin{definition}[$\be$-reduction]\label{def:beta-reduction}
	The \emph{$\be$-reduction} relation $\redb$ is the smallest relation on terms
	containing $(\la x:A.u)v \redb u[x/v]$, for all terms $A,u$ and $v$,
	and closed by subterm rewriting.
\end{definition}

\begin{definition}[$\Ga$-reduction]\label{def:gamma-reduction}
	Let $\Ga$ be a global context.
	The \emph{$\Ga$-reduction} relation $\redg$ is the smallest relation on terms
	containing $u \redg v$ for every rewrite rule $(u\rw v)\in\Ga$,
	closed by substitution and by subterm rewriting.
	We say that $\redg$ is \emph{left-algebraic} (respectively \emph{left-linear})
	if the rewrite rules in $\Ga$ are left-algebraic~(respectively left-linear).
\end{definition}

\begin{notation}
	We write $\redbg$ for $\redb\cup\redg$,
	$\eqb$ for the congruence generated by $\redb$
	and $\eqbg$ the congruence generated by $\redbg$.
\end{notation}

It is important to notice that these notions of reduction are defined as relations on all (untyped) terms.
In particular, we do not require the substitutions to be well-typed.
This allows defining the notion of rewriting independently from the notion of typing (see below).
This makes the system closer from what we would implement in practice.

Since the rewrite rules are either object-level or type-level, rewriting preserves the
three syntactic categories (object, type, kind).
Moreover, sorts are only convertible to themselves.

\subsection{Type System}
\label{ssec:lpm-type-system}
We now give the typing rules for the \lpm{}.
We begin by the inference rules for terms, then for local contexts and finally for global contexts.

\begin{figure}
	\begin{mdframed}
		\centering
		\begin{tabular}{lc}
			(\bf Sort) &
			\AxC{}
			\UIC{$\typg{\De}{\Type}{\Kind}$}
			\DP{}
			\\[0.5cm]

			(\bf Variable) &
			\AxC{$(x:A)\in\De$}
			\UIC{$\typg{\De}{x}{A}$}
			\DP{}
			\\[0.5cm]

			(\bf Constant) &
			\AxC{$(c:A)\in\Ga$}
			\UIC{$\typg{\De}{c}{A}$}
			\DP{}
			\\[0.5cm]

			(\bf Application) &
			\AxC{$\typg{\De}{t}{\Pi x:A.B}$}
			\AxC{$\typg{\De}{u}{A}$}
			\BIC{$\typg{\De}{tu}{B[x/u]}$}
			\DP{}
			\\[0.5cm]

			(\bf Abstraction) &
			\AxC{$\typg{\De}{A}{\Type}$}
			\AxC{$\typg{\De(x:A)}{t}{B}$}
			\AxC{$B\ne\Kind$}
			\TIC{$\typg{\De}{\la x:A.t}{\Pi x:A.B}$}
			\DP{}
			\\[0.5cm]

			(\bf Product) &
			\AxC{$\typg{\De}{A}{\Type}$}
			\AxC{$\typg{\De(x:A)}{B}{s}$}
			\BIC{$\typg{\De}{\Pi x:A.B}{s}$}
			\DP{}
			\\[0.5cm]

			(\bf Conversion) &
			\AxC{$\typg{\De}{t}{A}$}
			\AxC{$\typg{\De}{B}{s}$}
			\AxC{$A \eqbg B$}
			\TIC{$\typg{\De}{t}{B}$}
			\DP{}
		\end{tabular}
		\caption{Typing rules for terms in the \lpm.}\label{fig:typing-lpm}
	\end{mdframed}
\end{figure}

\begin{definition}[Well-Typed Term]\label{def:well-typed-term}
	We say that a term $t$ \emph{has type} $A$ in the global context $\Ga$ and the local context $\De$ if
	the judgment $\typg{\De}{t}{A}$ is derivable by the inference rules
	of \autoref{fig:typing-lpm}.
	We say that a term is \emph{well-typed} if such $A$ exists.
\end{definition}

The typing rules only differ from the usual typing rules for the $\la\Pi$-calculus by the {\bf (Conversion)}
rule where the congruence is extended from $\be$-conversion to $\be\Ga$-conversion allowing taking into
account the rewrite rules in the global context.

\begin{figure}
	\begin{mdframed}
		\centering
		\begin{tabular}{lc}
			(\bf Empty Local Context) &
			\AxC{}
			\UIC{$\typctx{\emptyset}$}
			\DP{}
			\\[0.5cm]

			(\bf Variable Declaration) &
			\AxC{$\typctx{\De}$}
			\AxC{$\typg{\De}{U}{\Type}$}
			\AxC{$x\notin dom(\De)$}
			\TIC{$\typctx{\De(x:U)}$}
			\DP{}
		\end{tabular}
		\caption{Typing rules for local contexts}\label{fig:typing-local-contexts}
	\end{mdframed}
\end{figure}

\begin{definition}[Well-Formed Local Context]\label{def:well-formed-contexts-lpm}
	A local context $\De$ is \emph{well-formed} with respect to a global context $\Ga$
	if the judgment $\typctx{\De}$ is derivable by the inference rules
	of \autoref{fig:typing-local-contexts}.
\end{definition}

Well-formed local contexts ensure that local declarations are unique and well-typed.

Besides the new conversion relation, the main difference between the $\la\Pi$-calculus and the \lpm{}
is the presence of rewrite rules in global contexts. We need to take this into account when
typing global contexts.

A key feature of any type system is the preservation of typing by reduction: the subject reduction property.
\begin{definition}[Subject Reduction]
    Let $\Ga$ be a global context.
	We say that a rewriting relation $\red$ satisfies the \emph{subject reduction} property in $\Ga$
	if, for all terms $t_1,t_2,T$ and local context $\De$ such that $\typctx{\De}$,
    $\typg{\De}{t_1}{T}$ and $t_1 \red t_2$ imply $\typg{\De}{t_2}{T}$.
\end{definition}
In the \lpm{}, we cannot allow adding arbitrary rewrite rules in the context, if we want to preserve subject reduction.
In particular, to prove subject reduction for the $\be$-reduction we need the following property:
\begin{definition}[Product-Compatibility]
	We say that a global context $\Ga$ satisfies the \emph{product compatibility} property
	(and we note ${\bf PC}(\Ga)$) if the following proposition is verified:\\
	if $\Pi x:A_1.B_1$ and $\Pi x:A_2.B_2$ are two
	well-typed product types in the same well-formed local context
	such that $\Pi x:A_1.B_1 \eqbg \Pi x:A_2.B_2$
	then $A_1 \eqbg A_2$ and $B_1 \eqbg B_2$.
\end{definition}

On the other hand, subject reduction for the $\Ga$-reduction requires rewrite rules to be well-typed
in the following sense:
\begin{definition}[Well-typed Rewrite Rules]\ 
	\begin{itemize}
		\item A rewrite rule $(u\rw v)$ is \emph{well-typed} for a global context $\Ga$
			if, for any substitution $\si$, well-formed local context $\De$ and term $T$,
			$\typg{\De}{\si(u)}{T}$ implies $\typg{\De}{\si(v)}{T}$.
		\item A rewrite rule is \emph{permanently well-typed} for a global context $\Ga$
            if it is well-typed
			for any extension $\Ga_0\supset\Ga$ that satisfies product compatibility.
			We write $\typrr{u}{v}$ when $(u\rw v)$ is permanently well-typed in $\Ga$.
	\end{itemize}
\end{definition}
The notion of permanently well-typed rewrite rule makes possible to typecheck rewrite rules
only once and not each time we make new declarations or add other rewrite rules in the context.

We can now give the typing rules for global contexts.

\begin{figure}
	\begin{mdframed}
		\centering
		\begin{tabular}{lc}
			(\bf Empty Global Context) &
			\AxC{}
			\UIC{$\wf{\emptyset}$}
			\DP{}
			\\[0.5cm]

			(\bf Object Declaration) &
			\AxC{$\wf{\Ga}$}
			\AxC{$\typg{\emptyset}{U}{\Type}$}
			\AxC{$\ud{c}\notin dom(\Ga)$}
			\TIC{$\wf{\Ga(\ud{c}:U)}$}
			\DP{}
			\\[0.5cm]

			(\bf Type Declaration) &
			\AxC{$\wf{\Ga}$}
			\AxC{$\typg{\emptyset}{K}{\Kind}$}
			\AxC{${\bf PC}(\Ga(C:K))$}
			\AxC{$C\notin dom(\Ga)$}
			\QIC{$\wf{\Ga(C:K)}$}
			\DP{}
			\\[0.5cm]

			(\bf Rewrite Rules) &
			\AxC{$\wf{\Ga}$}
			\AxC{$(\forall i)\typrr{u_i}{v_i}$}
			\AxC{${\bf PC}(\Ga(u_1\rw v_1)\ldots(u_n\rw v_n))$}
			\TIC{$\wf{\Ga(u_1\rw v_1)\ldots(u_n\rw v_n)}$}
			\DP{}
		\end{tabular}
		\caption{Typing rules for global contexts}\label{fig:typing-rules-for-global-contexts}
	\end{mdframed}
\end{figure}

\begin{definition}[Well-formed Global Context]\label{def:well-formed-global-context}
	A global context is \emph{well-formed} if the judgment $\wf{\Ga}$ is derivable
	by the inference rules of \autoref{fig:typing-rules-for-global-contexts}.
\end{definition}

The rules {\bf(Object Declaration)} and {\bf(Type Declaration)} ensure that constant declarations are
well-typed.
One can remark that the premise ${\bf PC}(\Ga(\ud{c}:U))$ is \emph{missing} in the
{\bf(Object Declaration)} rule.
This is because {\bf PC}($\Ga(\ud{c}:U)$) can be proved from {\bf PC}($\Ga$);
to prove product compatibility for $\Ga(c:U)$ it suffices to emulate the constant $c$ by a fresh variable
and use the product compatibility property of $\Ga$.
This cannot be done for type declarations since type-level variables do not exist in the \lpm{}.
The rule {\bf(Rewrite Rules)} permits adding rewrite rules.
Notice that we can add several rewrite rules at once.
In this case, only product compatibility for the whole system is required.
On the other hand, when a rewrite rule is added it needs to be well-typed independently
from the other rules that are added at the same time.

Well-formed global contexts satisfy subject reduction and uniqueness of types.
Proofs can be found in the long version of this paper at the author's webpage.

\begin{theorem}[Subject Reduction]\label{thm:SubjectReduction}
	Let $\Ga$ be a well-formed global context.
    Subject reduction holds for $\redbg$ in $\Ga$.
\end{theorem}

\begin{theorem}[Uniqueness of Types]\label{thm:UniquenessOfType}
	Let $\Ga$ be a well-formed global context
	and let $\De$ be a local context well-formed for $\Ga$.
	If $\typg{\De}{t}{T_1}$ and $\typg{\De}{t}{T_2}$
	then $T_1 \eqbg T_2$.
\end{theorem}

Remark that strong normalization of well-typed terms for the relations $\redg$ and $\redb$
is not guaranteed.

\subsection{Criteria for Product Compatibility and Well-typedness of Rewrite Rules}\label{ssec:criteria}

We now give effective criteria for checking product compatibility and well-typedness of rewrite rules.

The usual way to prove product compatibility is by showing the confluence of the rewrite system.
\begin{theorem}[Product Compatibility from Confluence]\label{thm:product-compat-by-confluence}
	Let $\Ga$ be a global context.
	If $\redbg$ is confluent then product compatibility holds for $\Ga$.
\end{theorem}

One could think that we can weaken the assumption of confluence requiring only
confluence for well-typed terms.
This is not a viable option since, without product compatibility,
we do not know if reduction preserves typing (subject reduction)
and if the set of well-typed terms is closed by reduction.
Therefore, it seems unlikely to be able to prove confluence only for
well-typed terms before proving the product compatibility property.

The confluence of $\redbg$ can be obtained from the confluence of $\redg$.
\begin{theorem}[M\"uller~\cite{Muller92}]\label{thm:Muller}
	If $\redg$ is left-algebraic, left-linear and confluent, then $\redbg$ is confluent.
\end{theorem}

To show that a rewrite rule is well-typed, one can use the following result:
\begin{theorem}\label{thm:strongly-well-typed}
	Let $\Ga$ be a well-formed global context and $(u\rw v)$ be a rewrite rule.
	If $u$ is algebraic and there exist $\De$ and $T$ such that
	$\typctx{\De}$, $dom(\De) = FV(u)$, $\typg{\De}{u}{T}$ and $\typg{\De}{v}{T}$
	then $(u\rw v)$ is permanently well-typed for $\Ga$.
\end{theorem}

\subsection{Example}
\label{ssec:lpm-example}

As an example, we define the map function on lists of integers.
We first define the type of \emph{Peano integers} by the three successive global declarations:\\[0.2cm]
$\Nat~:~\Type$.\\
$0~:~\Nat$.\\
$\Succ~:~\Nat \ar \Nat$.\\
For readability, we will write $n$ instead of $\overbrace{S~(S~\ldots (S}^{n~times}~0))$.
We now define a type for lists:\\[0.2cm]
$\List~:~\Type$.\\
$\Nil~:~\List$.\\
$\Cons~:~\Nat~\ar~\List~\ar~\List$.\\[0.2cm]
and the function map on lists:\\[0.2cm]
$\Map~:~(\Nat~\ar~\Nat)~\ar~\List~\ar~\List$.\\
$\Map~f~\Nil~\rw~\Nil$.\\
$\Map~f~(\Cons~hd~tl)~\rw~\Cons~(f~hd)~(\Map~f~tl)$.\\[0.2cm]
For instance, we can use this function to add some value to the elements of a list.
First, we define addition:\\[0.2cm]
$\Plus~:~\Nat \ar \Nat \ar \Nat$.\\
$\Plus~0~n~\rw~n$.\\
$\Plus~(\Succ~n_1)~n_2~\rw~\Succ~(\Plus~n_1~n_2)$.\\[0.2cm]
Then, we have the following reduction:\\[0.2cm]
$\Map~(\Plus~3)~(\Cons~1~(\Cons~2~(\Cons~3~\Nil))) \redg^* \Cons~4~(\Cons~5~(\Cons~6~\Nil))$.\\

This global context is well-formed.
Indeed, one can check that each global declaration is well-typed.
Moreover, each time we add a rewrite rule,
it verifies the hypotheses of~\autoref{thm:strongly-well-typed}
and it preserves the confluence of the relation $\redbg$.
Therefore, the rewrite rules are permanently well-typed and, by~\autoref{thm:product-compat-by-confluence},
product compatibility is always guaranteed.

\section{A Naive Definition of Rewriting Modulo \texorpdfstring{$\be$}{}}
\label{sec:naive-beta-rewriting}

As already mentioned, our goal is to give a notion of rewriting modulo $\be$ in the setting of \lpm{}.
We first exhibit the issues arising from a naive definition of this notion.

In an untyped setting, we could define rewriting modulo $\be$ in this manner:
$t_1$ rewrites to $t_2$ if, for some rewrite rule $(u \rw v)$ and substitution $\si$,
$\si(u) \eqb t_1$ and $\si(v) \eqb t_2$.
This definition is not satisfactory for several reasons.
\paragraph{It breaks subject reduction.}
For the rewrite rule of \autoref{sec:introduction},
taking $\si=\{f\mapsto \la y:\Omega.y\}$ where $\Omega$ is some ill-typed term, we have\\[-0.3cm]
$$ \Diff~(\la x:R.\Exp~x) \ar \Fmult~(\Diff~(\la x:R.(\la y:\Omega.y)~x)~(\la x:R.\Exp~((\la y:\Omega.y)~x)))$$
and, even if $\Diff~(\la x:R.\Exp~x)$ is well-typed, its reduct is ill-typed since it contains
an ill-typed subterm.
\paragraph{It may introduce free variables.}
In the example above, $\Omega$ has no reason to be closed.
\paragraph{It does not provide confluence.}
If we consider the following variant of the rewrite rule
$$ \Diff~(\la x:R.\Exp~(f~x)) \rw \Fmult~(\Diff~f)~(\la x:R.\Exp~(f~x)) $$
and take $\si_1 = \{ f \mapsto \la y:A_1.y \}$ and $\si_2 = \{ f \mapsto \la y:A_2.y \}$
where $A_1$ and $A_2$ are two non convertible types then we have:

\begin{tikzpicture}
	\node (n1) at (4,2) {$\Diff~(\la x:R.\Exp~((\la y:R.y)~x))$};
	\node (n2) at (0,0) {$\Fmult~(\Diff~(\la y:A_1.y))~(\la x:R.(\Exp~((\la y:A_1.y)~x)))$};
	\node (n3) at (9,0) {$\Fmult~(\Diff~(\la y:A_2.y))~(\la x:R.(\Exp~((\la y:A_2.y)~x)))$};

	\draw[->] (n1) -- (n2) node[pos=0.7,left]{$\Diff^{\si_1}$};
	\draw[->] (n1) -- (n3) node[pos=0.7,right]{$\begin{array}{l}\Diff^{\si_2}\\ \end{array}$};
\end{tikzpicture}
and the peak is not joinable.

Therefore, we need to find a definition that takes care of these issues.
We will achieve this using an embedding of \lpm{} into Higher-Order Rewrite Systems.

\section{Higher-Order Rewrite Systems}
\label{sec:higher-order-rewrite-systems}

In 1991, Nipkow~\cite{Nipkow-LICS-91} introduced Higher-Order Rewrite Systems (HRS)
in order to lift termination and confluence results from first-order rewriting
to rewriting over $\la$-terms. More generally, the goal was to study rewriting over terms
with bound variables such as programs, theorem and proofs.

Unlike the \lpm{}, in HRSs $\be$-reduction and rewriting do not operate at
the same level. Rewriting is defined as a relation between the $\be\eta$-equivalence
classes of simply typed $\la$-terms: the $\la$-calculus is used as a meta-language.

Higher-Order Rewrite Systems are based upon the (pre)terms of
the simply-typed $\la$-calculus built from a signature.
A signature is a set of base types $\Base$ and a set of typed constants.
A simple type is either a base type $b\in\Base$ or an arrow $A \ar B$
where $A$ and $B$ are simple types.

\begin{definition}[Preterm]\label{def:preterm}
	A \emph{preterm} of type $A$ is
	\begin{itemize}
		\item either a \emph{variable} $x$ of type $A$
			(we assume given for each simple type $A$ an infinite number of variables of this type),
		\item or a \emph{constant} $f$ of type $A$,
		\item or an \emph{application} $t(u)$ where
			$t$ is a preterm of type $B\ar A$ and
			$u$ is a preterm of type $B$,
		\item or, if $A=B\ar C$, an \emph{abstraction} $\ula x.t$ where
			$x$ is a variable of type $B$ and
			$t$ is a preterm of type $C$.
	\end{itemize}
\end{definition}

In order to distinguish the abstraction of HRSs from the abstraction of \lpm,
we use the underlined symbol $\ula$ instead of $\la$.
Similarly, we write the application $t(u)$ for HRSs (instead of $tu$).
We use the abbreviation $t(u_1,\ldots,u_n)$ for $t(u_1)\ldots(u_n)$.
If $A$ is a simple type, we write $A^1$ for $A$ and $A^{n+1}$ for $A \ar A^n$.

Notice also that HRSs abstractions do not have type annotations because variables are typed.

$\be$-reduction and $\eta$-expansion are defined as usual on preterms.
We write $\bel t$ for the long $\be\eta$-normal form of $t$.

\begin{definition}[Term]
	A \emph{term} is a preterm in long $\be\eta$-normal form.
\end{definition}

\begin{definition}[Pattern]\label{def:pattern}
	A term $t$ is a \emph{pattern} if every free occurrence of a variable $F$ is in
	a subterm of $t$ of the form $F\vec{u}$ such that $\vec{u}$ is $\eta$-equivalent to a list
	of distinct bound variables.
\end{definition}

The crucial result about patterns (due to Miller~\cite{Miller91alogic})
is the decidability of higher-order unification
(unification modulo $\be\eta$) of patterns.
Moreover, if two patterns are unifiable then a most general unifier exists and is computable.

The notion of rewrite rule for HRSs is the following:

\begin{definition}[Rewrite Rules]\label{def:HRS-rewrite-rule}
	A \emph{rewrite rule} is a pair of terms $(l \rw r)$ such that
	$l$ is a pattern not $\eta$-equivalent to a variable,
	$FV(r) \subset FV(l)$
	and $l$ and $r$ have the same base type.
\end{definition}

The restriction to patterns for the left-hand side ensures that matching is decidable but also that,
when it exists, the resulting substitution is unique.
This way, the situation is very close to first-order (\emph{i.e.} syntactic) matching.

\begin{definition}[Higher-Order Rewriting System (HRS)]\label{def:HRS}
	A \emph{Higher-Order Rewriting System} is a set $R$ of rewrite rules.

	The rewrite relation $\red_R$ is the smallest relation on terms closed by subterm rewriting
	such that, for any $(l \rw r)\in R$ and any well-typed substitution $\si$,
	$\bel \si(l) \red_R \bel \si(r)$.
\end{definition}

The standard example of an HRS is the untyped $\la$-calculus.
The signature involves a single base type \hrsTerm{} and two constants:
$$ \hrsLam : (\hrsTerm \ar \hrsTerm) \ar \hrsTerm $$
$$ \hrsApp : \hrsTerm \ar \hrsTerm \ar \hrsTerm $$
and a single rewrite rule for $\be$-reduction:
$$ (beta)\quad\hrsApp(\hrsLam(\ula x.X(x)),Y) \rw X(Y) $$

\section{An Encoding of the \texorpdfstring{\lpm{}}{} into Higher-Order Rewrite Systems}
\label{sec:encoding}

\subsection{Encoding of Terms}

We now mimic the encoding of the untyped $\la$-calculus as an HRS and encode the terms of the \lpm{}.
First we specify the signature.

\begin{definition}\label{def:SignatureLpm}
	The signature $\Sig(\la\Pi)$ is composed of
	a single base type $\hrsTerm$,
	the constants \hrsType{} and \hrsKind{} of atomic type \hrsTerm{},
	the constant  \hrsApp{} of type \hrsTerm{} $\ar$ \hrsTerm{} $\ar$ \hrsTerm{},
	the constants \hrsLam{} and \hrsPi{}
	of type \hrsTerm{} $\ar$ (\hrsTerm{} $\ar$ \hrsTerm{}) $\ar$ \hrsTerm{}
	and the constants {\tt c} of type \hrsTerm{}
		for every constant $c\in\OCst\cup\TCst$.
\end{definition}

Then we define the encoding of $\la\Pi$-terms.

\begin{definition}[Encoding of $\la\Pi$-term]\label{def:encoding-terms}
	The function $\emb{.}$ from $\la\Pi$-terms to HRS-terms in the signature $\Sig(\la\Pi)$
	is defined as follows:

		\begin{tabular}{lcllcl}
			$\emb{\Kind}$		&:= & \hrsKind &
			$\emb{\Type}$ 		&:= & \hrsType\\
			$\emb{x}$			&:= & $x$ (variable of type \hrsTerm) &
			$\emb{c}$			&:= & {\tt c} \\
			$\emb{uv}$			&:= & \hrsApp$(\emb{u},\emb{v})$ &
			$\emb{\la x:A.t}$ 	&:= & \hrsLam$(\emb{A},\ula x.\emb{t})$ \\
			$\emb{\Pi x:A.B}$ 	&:= & \hrsPi$(\emb{A},\ula x.\emb{B})$ & & &
		\end{tabular}
\end{definition}

\begin{lemma}\label{rmk:Bijection}
	The function $\emb{.}$ is a bijection from the $\la\Pi$-terms to HRS-terms of type \hrsTerm.
\end{lemma}

Note that this is a bijection between the untyped terms of the \lpm{} and well-typed terms of the corresponding HRS.
\subsection{Higher-Order Rewrite Rules}
\label{ssec:HO-rules}

We have faithfully encoded the terms. The next step is to encode the rewrite rules.
The following rule corresponds to $\be$-reduction at the HRS level:
$$ (beta)\quad\hrsApp(\hrsLam(X,\ula x.Y(x)),Z) \rw Y(Z) $$
We have the following correspondence:

\begin{lemma}\label{rmk:BetaBeta}\ 
	\begin{itemize}
		\item If $t_1 \redb t_2$ then $\emb{t_1} \red_{(beta)} \emb{t_2}$.
		\item If $t_1 \red_{(beta)} t_2$ and $t_1,t_2$ have type \hrsTerm{} then
			$\emb{t_1}^{-1} \redb \emb{t_2}^{-1}$ (where $\emb{.}^{-1}$ is the inverse of $\emb{.}$).
	\end{itemize}
\end{lemma}

By encoding rewrite rules in the obvious way (translating $(u\rw v)$ by $(\emb{u}\rw\emb{v})$), we would get
a similar result for $\Ga$-reduction.
But, since we want to incorporate rewriting modulo $\be$,
we proceed differently.

First, we introduce the notion of uniform terms.
These are terms verifying an arity constraint on their free variables.

\begin{definition}[Uniform Terms]
	A term $t$ is \emph{uniform} for a set of variables $V$ if
	all occurrences of a variable free in $t$ not in $V$ is applied to the same
	number of arguments.
\end{definition}

\newcommand\Pat{\mathcal{P}}
\newcommand\ari{\mathcal{A}}

Now, we define an encoding for uniform terms.

\begin{definition}[Encoding of uniform terms]\label{def:mod:encoding-uniform-term}
	Let $V$ be a set of variables and
	$t$ be a term uniform in $V$.
	The HRS-term $\embfv{u}{V}$ of type \hrsTerm{} is defined as follows:\\[0.2cm]
	\begin{tabular}{lcl}
		$\embfv{\Kind}{V}$		&:= & \hrsKind \\
		$\embfv{\Type}{V}$ 		&:= & \hrsType\\
		$\embfv{x}{V}$			&:= & $x$ \emph{if $x\in V$ (variable of type \hrsTerm)} \\
		$\embfv{c}{V}$			&:= & {\tt c} \\
		$\embfv{\la x:A.u}{V}$ 	&:= & \hrsLam($\embfv{A}{V}$, $\ula x.\embfv{u}{V\cup\{x\}})$ \\
		$\embfv{\Pi x:A.B}{V}$ 	&:= & \hrsPi($\embfv{A}{V}$, $\ula x.\embfv{B}{V\cup\{x\}}$) \\
		$\embfv{x\vec{v}}{V}$	&:= &
			$x(\embfv{\vec{v}}{V})$ \emph{if $x\notin V$ ($x$ of type $\hrsTerm^{n+1}$ where $n=|\vec{v}|$)} \\
		$\embfv{uv}{V}$	&:= &
			\hrsApp($\embfv{u}{V}$,$\embfv{v}{V}$) \emph{if $uv \ne x~\vec{w}$ for $x\notin V$}
	\end{tabular}
\end{definition}

Now, we define an equivalent of patterns for the \lpm{}.

\begin{definition}[$\la\Pi$-patterns]\label{def:la-Pi-pattern}
	Let $V_0$ be a set of variables, $\ari$ be a function giving an arity to variables and let $V=(V_0,\ari)$.
	The subset $\Pat_V$ of $\la\Pi$-terms is defined inductively as follows:
	\begin{itemize}
		\item if $c$ is a constant, then $c\in\Pat_V$;
		\item if $p,q\in\Pat_V$, then $p~q\in\Pat_V$;
		\item if $x\in V_0$, then $x\in\Pat_V$;
		\item if $p\in\Pat_V$, $x\notin V_0$ and $\vec{y}$ is a vector
			of pairwise distinct variables in $V_0$ such that $|\vec{y}|=\ari(x)$,
			then $p~(x~\vec{y})\in\Pat_V$;
		\item if $p\in\Pat_V$, $FV(A)\subset V_0$ and $q\in\Pat_{(V_0\cup\{x\},\ari)}$, then
			$p~(\la x:A.q)\in\Pat_V$;
	\end{itemize}
	A term $t$ is a \emph{$\la\Pi$-pattern} if, for some arity function $\ari$, $t\in\Pat_{(\emptyset,\ari)}$.
\end{definition}

Remark that the encoding of a $\la\Pi$-pattern as a uniform term is a pattern.

We now define the encoding of rewrite rules.

\begin{definition}[Encoding of Rewrite Rules]\label{def:encoding-rules}
	Let $(u\rw v)$ be a rewrite rule such that
	\begin{itemize}
		\item $u$ is a $\la\Pi$-pattern;
		\item $FV(v)\subset FV(u)$;
		\item all free occurrences of a variable in $u$ and $v$ are applied to
			the same number of arguments.
	\end{itemize}
	The encoding of $(u\rw v)$ is
	$\emb{u\rw v} = \embfv{u}{\emptyset}\rw\embfv{v}{\emptyset}$.
\end{definition}

Remark that the first assumption ensures that the left-hand side is a pattern
and the third assumption ensures that the HRS-term is well-typed.

\begin{definition}[HRS($\Ga$)]\label{def:HRSGa}
	Let $\Ga$ a global context whose rewrite rules satisfy the condition of \autoref{def:encoding-rules}.
	We write HRS($\Ga$) for the HRS $\{ \emb{u\rw v}~|~(u\rw v)\in\Ga\}$
	and HRS($\be\Ga$) for $HRS(\Ga)\cup\{(beta)\}$.
\end{definition}

\section{Rewriting Modulo \texorpdfstring{$\be$}{}}
\label{sec:rewriting-modulo-beta}

\subsection{Definition}
\label{ssec:beta-rewriting-definition}

We are now able to properly define rewriting modulo $\be$.
As for usual rewriting, rewriting modulo $\be$ is defined on all (untyped) terms.

\begin{definition}[Rewriting Modulo $\be$]\label{def:rewriting-modulo}
	Let $\Ga$ be a global context.
	We say that $t_1$ \emph{rewrites to $t_2$ modulo $\be$} (written $t_1 \redgb t_2$)
	if $\emb{t_1}$ rewrites to $\emb{t_2}$ in HRS($\Ga$).
	Similarly, we write $t_1 \redbgb t_2$
	if $\emb{t_1}$ rewrites to $\emb{t_2}$ in HRS($\be\Ga$).
\end{definition}

\begin{lemma}\label{rmk:RedbgRedbgb}\ 
	\begin{itemize}
		\item $\redbgb = \redgb \cup \redb$.
		\item If $t_1 \redg t_2$ then $t_1 \redgb t_2$.
	\end{itemize}
\end{lemma}

\subsection{Example}
\label{ssec:rewriting-modulo-example}

Let us look at the example from the introduction.
Now we have :
$$\Diff~(\la x:R.\Exp~x) \redgb \Fmult~(\Diff~(\la x:R.x))~(\la x:R.\Exp~x) $$
Indeed, for $\si = \{ f \mapsto \ula y.y\}$ we have
$$\emb{\Diff~(\la x:R.\Exp~x)}
	= \hrsApp (\Diff, \hrsLam(R,\ula x.\hrsApp(\Exp,x)))
	= \bel \si( \hrsApp (\Diff, \hrsLam(R,\ula x.\hrsApp(\Exp,f(x))) ))$$
and
$$
\begin{array}{ll}
	\emb{\Fmult~(\Diff~(\la x:R.x))~(\la x:R.\Exp~x)}
	& = \hrsApp ( \Fmult , \hrsApp ( \Diff, \hrsLam (R,\ula x.x) ), \hrsLam (R, \ula x.\hrsApp(\Exp,x) )) \\
	& = \bel \si( \hrsApp ( \Fmult , \hrsApp ( \Diff, \hrsLam (R, \ula x.f(x)) ) ,
	\hrsLam( R, \ula x. \hrsApp ( \Exp, f(x) ) ) ) )
\end{array}
$$

Therefore, the peak is now joinable.

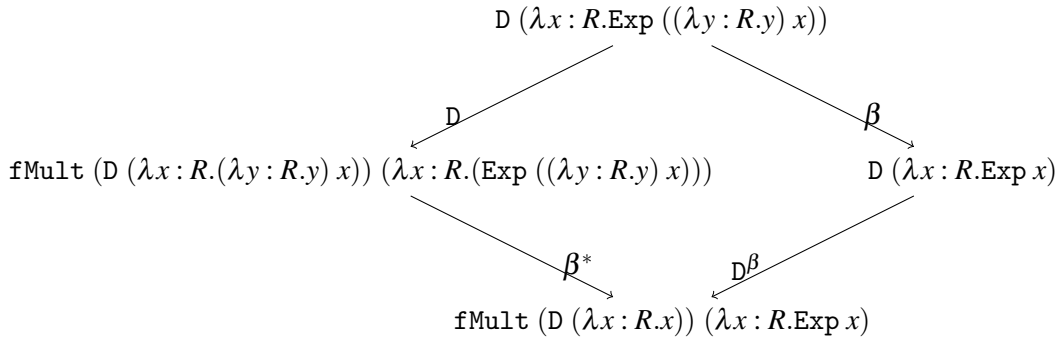
\begin{tikzpicture}
	\node (n1) at (8,2) {$\Diff~(\la x:R.\Exp~((\la y:R.y)~x))$};
	\node (n2) at (4,0) {$\Fmult~(\Diff~(\la x:R.(\la y:R.y)~x))~(\la x:R.(\Exp~((\la y:R.y)~x)))$};
	\node (n3) at (12,0) {$\Diff~(\la x:R.\Exp~x)$};
	\node (n4) at (8,-2) {$\Fmult~(\Diff~(\la x:R.x))~(\la x:R.\Exp~x)$};

	\draw[->] (n1) -- (n2) node[pos=0.7,left]{\Diff};
	\draw[->] (n1) -- (n3) node[pos=0.7,right]{$\be$};
	\draw[->] (n3) -- (n4) node[pos=0.7,left]{$\Diff^\be$};
	\draw[->] (n2) -- (n4) node[pos=0.7,right]{$\be^*$};
\end{tikzpicture}

In fact the rewriting relation can be shown confluent~\cite{DBLP:conf/hoa/Oostrom95}.

\subsection{Properties}
\label{ssec:properties-modulo}

Rewriting modulo $\be$ also preserves typing.

\begin{theorem}[Subject Reduction for $\redgb$]\label{thm:subject-reduction-modulo}
	Let $\Ga$ a well-formed global context and $\De$ a local context well-formed for $\Ga$.
	If $\typg{\De}{t_1}{T}$ and $t_1 \redgb t_2$ then $\typg{\De}{t_2}{T}$.
\end{theorem}

It directly follows from the following lemma:

\begin{lemma}\label{lem:RedgbImpliesRedg}
	If $t_1 \redgb t_2$ then, for some $t_1'$ and $t_2'$,
	we have $t_1 \leftarrow^*_\be t_1' \redg t_2' \redb^* t_2$.
	Moreover, if $t_1$ is well-typed then we can choose $t_1'$
	such that it is well-typed in the same context.
\end{lemma}
\begin{proof}
	The idea is to lift the $\be$-reductions that occur at the HRS level to the \lpm{}.
	Suppose $t_1 \redgb t_2$.
	For some rewrite rule $(u \rw v)$ and (HRS) substitution $\si$,
	we have $\bel \si(u) = \emb{t_1}$ and $\bel \si(v) = \emb{t_2}$.
	We define the ($\la\Pi$) substitution $\hat{\si}$ as follows:
	$\hat{\si}(x) = \emb{\si(x)}^{-1}$ if $\si(x)$ has type $\hrsTerm$;
	$\hat{\si}(x) = \la\vec{x}:\vec{A}. \emb{u}^{-1}$
	if $\si(x) = \ula \vec{x}.u$ has type $\hrsTerm^n \ar \hrsTerm$
	where the $A_i$ are arbitrary types.
	We have, at the $\la\Pi$ level,
	$\hat{\si}(u) \redg \hat{\si}(v)$,
	$\hat{\si}(u) \redb^* t_1$ and
	$\hat{\si}(v) \redb^* t_2$.
	If $t_1$ is well-typed then the $A_i$ can be chosen so that
	$\hat{\si}(u)$ is also well-typed.
\end{proof}

Another consequence of this lemma is that the rewriting modulo $\be$ does not modify the congruence.

\begin{theorem}\label{thm:eqbgIseqbgb}
	The congruence generated by $\redbgb$ is equal to $\eqbg$.
\end{theorem}
\begin{proof}
	Follows from~\autoref{rmk:RedbgRedbgb} and~\autoref{lem:RedgbImpliesRedg}.
\end{proof}

\subsection{Generalized Criteria for Product Compatibility and Well-Typedness of Rewrite Rules}
\label{ssec:criteria-modulo}

Using our new notion of rewriting modulo $\be$, we can generalize the criteria of~\autoref{ssec:criteria}.

\begin{theorem}\label{thm:product-compat-by-HO-confluence}
	Let $\Ga$ be a global context. If HRS($\be\Ga$) is confluent, then product compatibility holds for~$\Ga$.
\end{theorem}
\begin{proof}
	Assume that $\Pi x:A_1.B_1 \eqbg \Pi x:A_2.B_2$ then,
	by~\autoref{thm:eqbgIseqbgb}, $\Pi x:A_1.B_1 \eqbgb \Pi x:A_2.B_2$.
	By confluence, there exist $A_0$ and $B_0$ such that
	$A_1 \redbgb^* A_0$, $A_2 \redbgb^* A_0$, $B_1 \redbgb^* B_0$ and $B_2 \redbgb^* B_0$.
	It follows, by~\autoref{thm:eqbgIseqbgb}, that $A_1 \eqbg A_2$ and $B_1 \eqbg B_2$.
\end{proof}

To prove the confluence of a HRS,
one can use van Oostrom's development-closed theorem~\cite{DBLP:conf/hoa/Oostrom95}.

\autoref{thm:strongly-well-typed} can also be generalized to deal with $\la\Pi$-patterns.

\begin{theorem}\label{thm:strong-well-typed-pattern}
	Let $\Ga$ be a well-formed global context and $(u\rw v)$ be a rewrite rule.
	If $u$ is a $\la\Pi$-pattern and there exist $\De$ and $T$ such that
	$\typctx{\De}$, $FV(u)=dom(\De)$, $\typg{\De}{u}{T}$ and $\typg{\De}{v}{T}$
	then $(u\rw v)$ is permanently well-typed for $\Ga$.
\end{theorem}

This theorem is a corollary of the following lemma.

\newcommand\eqbgdeux{\equiv_{\be\Ga_2}}

\begin{lemma}\label{lem:MainLemma}
	Let $\Ga\subset\Ga_2$ be two well-formed global contexts.
	If $t\in\Pat_{dom(\Si)}$,
	$dom(\si)=dom(\De)$,
	for all $(x:A)\in\Si$, $\si(A)=A$,
	$\typg{\De\Si}{t}{T}$ and
	$\typ{\Ga_2}{\De_2\Si}{\si(t)}{T_2}$
	then
	$T_2 \eqbgdeux \si(T)$
	and, for all $x\in FV(t)\cap dom(\De)$,
		$\typ{\Ga_2}{\De_2}{\si(x)}{T_x}$ for $T_x \eqbgdeux \si(\De(x))$.
\end{lemma}
\begin{proof}
	We proceed by induction on $t\in\Pat_{dom(\Si)}$.
	\begin{itemize}
		\item if $t=c$ is a constant, then
			$FV(t)=\emptyset$ and, by inversion on $\typg{\De\Si}{t}{T}$, 
			there exists a (closed term) $A$ such that $(c:A)\in\Ga\subset\Ga_2$,
			$T \eqbg A$
			and $T_2 \eqbgdeux A$.
			Since $A = \si(A)$, we have $\si(T) \eqbgdeux T_2$.

		\item if $t=x\in dom(\Si)$, then, by inversion,
			there exists $A$ such that $(x:A)\in\Si$,
			$T \eqbg A$
			and $T_2 \eqbgdeux A$.
			Since $A = \si(A)$, we have $\si(T) \eqbgdeux T_2$.

		\item if $t=p~q$, then, by inversion,
			on the one hand,
			$\typg{\De\Si}{p}{\Pi x:A.B}$,
			$\typg{\De\Si}{q}{A}$
			and $T \eqbg B[x/q]$.
			On the other hand,
			$\typ{\Ga_2}{\De_2\Si}{\si(p)}{\Pi x:A_2.B_2}$,
			$\typ{\Ga_2}{\De_2\Si}{\si(q)}{A_2}$
			and $T_2 \eqbgdeux B_2[x/\si(q)]$.

			By induction hypothesis on $p$,
			we have $\si(\Pi x:A.B) \eqbgdeux \Pi x:A_2.B_2$ and
			for all $x\in FV(p)\cap dom(\De)$, $\typ{\Ga_2}{\De_2}{\si(x)}{T_x}$
				with $T_x \eqbgdeux \si(\De(x))$.

			By product-compatibility of $\Ga_2$, $\si(A) \eqbgdeux A_2$ and $\si(B) \eqbgdeux B_2$.
			It follows that $\si(T) \eqbgdeux \si(B[x/q]) \eqbgdeux B_2[x/\si(q)] \eqbgdeux T_2$.

			Now, we distinguish three sub-cases:
			\begin{itemize}
				\item either $q\in\Pat_{dom(\Si)}$ and
					by induction hypothesis on $q$,
					for all $x\in FV(q)\cap dom(\De)$,
					$\typ{\Ga_2}{\De_2}{\si(x)}{T_x}$ with $T_x \eqbgdeux \si(\De(x))$.
				\item Or $q=\la x:A.q_0$ with $FV(A)\in dom(\Si)$ and $q_0\in\Pat_{dom(\Si(x:A))}$ and
					by induction hypothesis on $q_0$,
					for all $x\in FV(q_0)\cap dom(\De)$,
					$\typ{\Ga_2}{\De_2}{\si(x)}{T_x}$ with $T_x \eqbgdeux \si(\De(x))$.
				\item Or $q=x\vec{y}$ with $x\notin dom(\Si)$ and $\vec{y}\subset dom(\Si)$.
					By inversion, on the one hand,
					$\De(x) \eqbg \Pi \vec{y}:\Si(\vec{y}).C$
					for $C \eqbg A$.
					On the other hand,
					$\typ{\Ga_2}{\De_2}{\si(x)}{\Pi \vec{y}:\Si(\vec{y}).C_2}$
					for $C_2 \eqbgdeux A_2$.
					Since $\si(A) \eqbgdeux A_2$, we have
					$\Pi \vec{y}:\Si(\vec{y}).C_2 \eqbgdeux \Pi \vec{y}:\Si(\vec{y}).\si(C) = \si(\De(x))$.
			\end{itemize}
	\end{itemize}
\end{proof}

\begin{proof}[Proof of~\autoref{thm:strong-well-typed-pattern}]
	Let $\Ga_2$ be a well-formed extension of $\Ga$.
	Suppose that $\typ{\Ga_2}{\De_2}{\si(u)}{T_2}$.

	By~\autoref{lem:MainLemma} and $FV(u)=dom(\De)$,
	we have, for all $x\in dom(\De)$, $\typ{\Ga_2}{\De_2}{\si(x)}{T_x}$
	for $T_x \equiv_{\be\Ga_2} \si(\De(x))$
	and $T_2 \equiv_{\be\Ga_2} \si(T)$.

	By induction on $\typg{\De}{v}{T}$, we deduce
	$\typ{\Ga_2}{\De_2}{\si(v)}{T_3}$, for $T_3 \equiv_{\be\Ga_2} \si(T) \equiv_{\be\Ga_2} T_2$.
	It follows, by conversion, that $\typ{\Ga_2}{\De_2}{\si(v)}{T_2}$.
\end{proof}

\section{Applications}
\label{sec:Applications}

\subsection{Parsing and Solving Equations}

\newcommand\Expr{{\tt expr}}
\newcommand\ExprS{{\tt expr\_S}}
\newcommand\ExprPlus{{\tt expr\_P}}
\newcommand\mkExpr{{\tt mk\_expr}}
\newcommand\parse{{\tt to\_expr}}
\newcommand\solve{{\tt solve}}
\newcommand\One{{\tt One}}
\newcommand\All{{\tt All}}
\newcommand\None{{\tt None}}
\newcommand\Solution{{\tt Solution}}

The context declarations and rewrite rules of~\autoref{fig:expression-parsing}
define a function \parse{} which parses a function of type $\Nat$ to $\Nat$ into
an expression of the form $a*x+b$ (represented by the term $\mkExpr~a~b$)
where $a$ and $b$ are constants.
The left-hand sides of the rewrite rules on \parse{} are $\la\Pi$-patterns.
This allows defining \parse{} by pattern matching in a way which looks under the binders.

The function \solve{} can then be used to solve
the linear equation $a*x+b=0$. The answer is either \None{} if there is no solution,
or \All{} if any $x$ is a solution or $\One~m~n$ if $-m/(n+1)$ is the only solution.

For instance, we have (writing $\One~-\frac{1}{3}$ for $\One~1~2$):
$$\solve~(\parse (\la x:\Nat. \Plus~x~(\Plus~x~(\Succ~x)) )) \redbg^* \One~-\frac{1}{3}.$$
\begin{figure}
	\begin{mdframed}
		\[
			\begin{array}{lll}
				\Expr	& :	& \Type.\\
				\mkExpr	& :	& \Nat \ar \Nat \ar \Expr.\\
				\ExprS	& : & \Expr \ar \Expr.\\
				\ExprS~(\mkExpr~a~b)	& \rw	& \mkExpr~a~(\Succ~b).\\
				\ExprPlus	& : & \Expr \ar \Expr \ar \Expr.\\
				\ExprPlus~(\mkExpr~a_1~b_1)~(\mkExpr~a_2~b_2)	& \rw	&
				\mkExpr~(\Plus~a_1~a_2)~(\Plus~b_1~b_2). \\[0.1cm]
				\parse	& :	& (\Nat \ar \Nat) \ar \Expr.\\
				\parse~(\la x:\Nat.0)					& \rw	& \mkExpr~0~0.\\
				\parse~(\la x:\Nat.\Succ~(f~x)) 		& \rw 	& \ExprS~(\parse~(\la x:\Nat.f~x)).\\
				\parse~(\la x:\Nat.x) 					& \rw 	& \mkExpr~(\Succ~0)~0.\\
				\parse~(\la x:\Nat.\Plus~(f~x)~(g~x))	& \rw	& \\
				\multicolumn{3}{r}{
					\ExprPlus~(\parse~(\la x:\Nat.f~x))~(\parse~(\la x:\Nat.g~x)).}\\[0.1cm]
				\Solution								& :		& \Type.\\
				\All									& :		& \Solution.\\
                \One									& :		& \Nat \ar \Nat \ar \Solution.\\
				\None									& :		& \Solution.\\
				\solve~(\mkExpr~0~0)					& \rw	&	\All.\\
				\solve~(\mkExpr~0~(S~n))				& \rw	&	\None.\\
				\solve~(\mkExpr~(S~n)~m)				& \rw	&	\One~m~n.
			\end{array}
		\]
	\caption{Parsing and solving linear equations}\label{fig:expression-parsing}
	\end{mdframed}
\end{figure}
By~\autoref{thm:product-compat-by-HO-confluence} and~\autoref{thm:strong-well-typed-pattern}
the global context of~\autoref{fig:expression-parsing} is well-formed.

\subsection{Universe Reflection}

In~\cite{Assaf15}, Assaf defines a version of the calculus of construction with explicit universe subtyping
thanks to an extended notion of conversion generated by a set of rewrite rules.
This work can easily be adapted to fit in the framework of the \lpm{}.
However, the confluence of the rewrite system holds only for rewriting modulo $\be$.

\section{Conclusion}
\label{sec:conclusion}

We have defined a notion of rewriting modulo $\be$ for the \lpm{}.
We achieved this by encoding the \lpm{} into the framework of Higher-Order Rewrite Systems.
As a consequence we also made available for the \lpm{} the confluence
criteria designed for the HRSs (see for instance~\cite{Nipkow-LICS-91}
or~\cite{DBLP:conf/hoa/Oostrom95}).
We proved that rewriting modulo $\be$ preserves typing.
We generalized the criterion for product compatibility, by replacing the assumption
of confluence by the confluence of the rewriting relation modulo $\be$.
We also generalized the criterion for well-typedness of rewrite rules to allow left-hand
to be $\la\Pi$-patterns.
These generalizations permit proving
subject reduction and uniqueness of types for more systems.

A natural extension of this work would be to consider rewriting
modulo $\be\eta$ as in Higher-Order Rewrite Systems.
This requires extending the conversion with $\eta$-reduction.
But, as remarked in~\cite{Geuvers92} (attributed to Nederpelt),
$\red_{\be\eta}$ is not confluent on untyped terms as the following example shows:
$$ \la y:B.y \leftarrow_\eta \la x:A.(\la y:B.y)x \redb \la x:A.x $$
Therefore properties such as product compatibility need to be proved another way.
We leave this line of research for future work.

For the $\la\Pi$-calculus a notion of higher-order pattern matching has been
proposed~\cite{DBLP:conf/popl/Pientka08} based on Contextual Type Theory
(CTT)~\cite{DBLP:journals/tocl/NanevskiPP08}.
This notion is similar to our.
However, it is defined using the notion of meta-variable (which is native in CTT)
instead of a translation into HRSs.

In~\cite{Blanqui15}, Blanqui studies the termination of the combination of $\be$-reduction
with a set of rewrite rules with matching modulo $\be\eta$ in the polymorphic $\la$-calculus.
His definition of rewriting modulo $\be\eta$ is direct and does not use any encoding.
This leads to a slightly different notion a rewriting modulo $\be$.
For instance, $\Diff (\la:R.\Exp~x)$ would reduce
to $\Fmult~(\Diff~(\la x:R.(\la y:R.y)~x))~(\la x:R.\Exp~((\la y:R.y)~x))$ instead
of $\Fmult~(\Diff~(\la x:R.x))~(\la x:R.\Exp~x)$.
It would be interesting to know whether the two definitions are equivalent with respect to confluence.

We implemented rewriting modulo $\be$ in Dedukti~\cite{Dedukti}, our type-checker for the \lpm.

\paragraph{Acknowledgments.}
The author thanks very much Ali Assaf, Olivier Hermant, Pierre Jouvelot and the reviewers for their very
careful reading and many suggestions.

\bibliographystyle{eptcs}
\bibliography{lfmtp15}


\end{document}